\documentclass [6pt,a4paper]{article}
\usepackage{bbm}
\usepackage{amsfonts}
\usepackage{mathrsfs}
\usepackage {amssymb}
\usepackage {amsmath}
\usepackage {ntheorem}
\usepackage{latexsym}
\usepackage{booktabs}
\usepackage{mathrsfs}
\usepackage{bm}
\usepackage{float}
\usepackage{color}
\usepackage{makecell}
\newenvironment{proof}{{\noindent\it\textbf{Proof}}\quad}{\hfill $\square$\par}
\usepackage{geometry}

\def\dse#1{\vskip 0.6cm\noindent
        {\large\bf #1}
        \vskip 0.4cm}
\def\dse#1{\vskip 0.6cm\noindent
        {\large\bf #1}
        \vskip 0.4cm}
\usepackage{lastpage}
\usepackage{epsfig}

\begin{document}\large
\newtheorem{lemma}{Lemma}[section]
\newtheorem{theorem}[lemma]{Theorem}
\newtheorem{example}[lemma]{Example}
\newtheorem{definition}[lemma]{Definition}
\newtheorem{proposition}[lemma]{Proposition}
\newtheorem{conjecture}[lemma]{Conjecture}
\newtheorem{corollary}[lemma]{Corollary}
\newtheorem{remark}{Remark}
%********************************************************************************************************
\begin{center}
\textbf{\LARGE{Cyclic codes and some new entanglement-assisted quantum MDS codes}}\footnote { E-mail
addresses: jiangw000@163.com(W.Jiang), zhushixin@hfut.edu.cn(S.Zhu), chenxiaojing0909@ahu.\\edu.cn(X.Chen).
}\\
\end{center}

\begin{center}
{ { Wan Jiang$^1$, \  Shixin Zhu$^{1}$, \  Xiaojing Chen$^2$} }
\end{center}

\begin{center}
\textit{1\ School of Mathematics, Hefei University of
Technology, Hefei 230601, Anhui, P.R.China\\
2\  School of Internet, Anhui University, Hefei 230039, Anhui, P.R.China}
\end{center}

\noindent\textbf{Abstract:} Entanglement-assisted quantum error correcting codes (EAQECCs) play a significant role in protecting quantum information from decoherence and quantum noise. Recently, constructing entanglement-assisted quantum maximum distance separable (EAQMDS) codes with flexible parameters has received much attention. In this work, four families of EAQMDS codes with a more general length are presented. And the method of selecting defining set is different from others. Compared with all the previously known results, the EAQMDS codes we constructed have larger minimum distance. All of these EAQMDS codes are new in the sense that their parameters are not covered by the quantum codes available in the literature.\\

\noindent\textbf{Keywords}:\ EAQECCs, EAQMDS codes, Cyclic codes, Defining set \     \\   % the keywords
\section{Introduction}

    Quantum error correcting codes (QECCs) play an important role in quantum computing and quantum communications \cite{ref1}-\cite{ref2}. Usually, we use $[[n, k, d]]_q$ to denote a $q$-ary quantum error correcting code (QECC), whose length is $n$ with size $q^k$ and minimum distance $d$. It can detect up to $d-1$ quantum errors and correct up to $\lfloor \frac{d-1}{2}\rfloor$ quantum errors. Similar to classical linear codes, there exist quantum Singleton bound $k\leq n-2d+2$. If $k=n-2d+2$, the QECC is called a quantum maximum-distance separable (QMDS) code. As we all know, one of the central topics in quantum coding theory is to construct quantum codes with good parameters, especially QMDS codes. After Calderbank et al. \cite{ref3} found that QECCs can be constructed from classical self-orthogonal codes with certain inner product, there are many good works about QECCs and QMDS codes \cite{ref4}-\cite{ref10}.

    The CSS construction and Hermitian construction are two famous construction methods for QECCs. However, they both need classical codes to be dual-containing or self-orthogonal, which is not easy to satisfy. This limitation condition form a gap between classical linear codes and QECCs. In 2006, Brun et al. \cite{ref11} proposed the concept of EAQECCs and solved this problem successfully. They proved that EAQECCs allows  non-dual-containing classical codes to construct QECCs if the sender and receiver shared entanglement in advance. 

    Let $q$ be a prime power. A $q$-ary EAQECC, denoted as $[[n, k, d; c]]_q$, that encodes $k$ information
    qubits into $n$ channel qubits with the help of $c$ pairs of maximally entangled states and can correct up to $\lfloor \frac{d-1}{2} \rfloor$ errors, where $d$ is the minimum distance of the code. Similar to QECCs,
    there is entanglement-assisted quantum Singleton bound in the following proposition.
  
 \begin{proposition}\label{pro:1.1}  \rm\cite{ref11,ref12} 
	\emph{Assume that $\mathcal{C}$ is an entanglement-assisted quantum code with parameters $[[n, k, d; c]]_q$. If $d\leq (n+2)/2 $, then $\mathcal{C}$ satisfies the entanglement-assisted Singleton bound
		 $$n+c-k\geq 2(d-1).$$
		  If $\mathcal{C}$ satisfies the equality $n+c-k= 2(d-1)$ for $d\leq (n+2)/2$, then it is called an entanglement-assisted quantum MDS code.}
\end{proposition}

 Although it is possible to construct EAQECCs from any classical linear codes in theory, it is not easy to calculate the number of entangled states $c$. This problem gets solved after Li et al. proposed the concept of decomposing the defining set \cite{ref13}. Besides, they used this method to construct some EAQECCs with good parameters \cite{ref14}. After that, this concept is generated to more general cases, from cyclic codes to negacyclic codes and constacyclic codes \cite{ref15}-\cite{ref16}. As EAQMDS codes are an important class of quantum codes, a great deal of efforts have been made to construct some new EAQMDS codes in recent years. There are many new EAQMDS codes with a small number of entangled states $c$ have been constructed \cite{ref18}-\cite{ref24}. 
     
    Actually, if we want to correct more quantum errors, the EAQECCs we constructed need to have larger minimum distance. And the larger the minimum distance of EAQECCs are, the more the entangled states $c$ will be employed. However, it is not an easy task to analyze the accurate parameters if the value of entangled states $c$ is too large or flexible. Recently, some scholars have obtained great progress. Luo et al. obtained several new infinite families of EAQECCs with
    flexible parameters by using generalized Reed–Solomon codes and
    extended generalized Reed–Solomon codes \cite{ref6}. In \cite{ref25}, Qian and Zhang constructed some new EAQMDS codes with length $n = q^2+1$ and some new entanglement-assisted quantum almost MDS codes. Besides, Wang et al. got series of EAQECCs with flexible parameters of length $n = q^2+1$ by using constacyclic codes \cite{ref26}.

    Based on the previous work, this paper dedicates to study cyclic codes to construct EAQECCs of length $n=(q^2+1)/a$ with flexible parameters naturally, where $a=m^2+1$ ($m\geq1$~is~odd). More precisely, we construct four classes of $q$-ary EAQMDS codes with parameters as follows:
    
    (1) $[[n, n-4\alpha(q-m-a\alpha)-4mk, 2(\alpha q+mk)+1; 4\alpha(a\alpha+m)]]_q,$ where $q=2ak+m$ $(k\geq 1)$ is an odd prime power and $1\leq \alpha\leq k$.

    (2) $[[n, n-4\alpha(q-a-m-a\alpha)-4(a+m)k-(a+2m), 2[\alpha q +(a+m)k+\frac{a+2m}{2}]+1; 4\alpha(a\alpha+a+m)+a+2m]]_q,$ where $q=2ak+a+m$ $(k\geq 1)$ is an odd prime power and $1\leq \alpha\leq k$.

    (3) $[[n, n-4\alpha(q-a+m-a\alpha)-4(a-m)k-(a-2m), 2[\alpha q +(a-m)k+\frac{a-2m}{2}]+1;
   4\alpha(a\alpha+a-m)+a-2m]]_q$, where $q=2ak+a-m$ $(k\geq 1)$ is an odd prime power and $1\leq \alpha\leq k$.

    (4) $[[n, n-4\alpha[q-(2a-m)-a\alpha]-4(2a-m)k-4(a-m), 2[\alpha q+(2a-m)k+2(a-m)]+1; 4\alpha(a\alpha +2a-m)+4(a-m)]]_q,$ where $q=2ak+2a-m$ $(k\geq 1)$ is an odd prime power and $1\leq \alpha\leq k$.
    
    The main organization of this paper is as follows. In Sect.2, some basic background and results about cyclic codes and EAQECCs are reviewed. In Sect.3, four classes of EAQMDS codes with length $n=(q^2+1)/a$ are obtained, where $a=m^2+1$ ($m\geq1$~is~odd). Sect.4 concludes the paper.

\section{Preliminaries} 
    
    In this section, we will review some relevant concepts and basic theory on cyclic codes and EAQECCs.
    
    Let $\mathbb{F}_{q^2}$ be a finite field of $q^2$ elements and $\mathbb{F}_{q^2}^n$ be the $n$-dimensional row vector space over $\mathbb{F}_{q^2}$, where $q$ is a prime power, $k$ and $n$ are positive integer. Let $\mathcal{C}$ be a $q^2$-ary $[n,k,d]$ linear code of length $n$ with dimension $k$ and minimum distance $d$, which is a $k$-dimensional linear subspace of $\mathbb{F}_{q^2}^n$. Then $\mathcal{C}$ satisfies the following Singleton bound:
    $$n-k \geq d-1.$$ When the equality $n-k=d-1$ holds, then $\mathcal{C}$ is called an maximal distance separable (MDS) code.
    Let vectors $\textbf {x} = (x_0,x_1,\cdots,x_{n-1})$ and $\textbf {y} = (y_0,y_1,\cdots,y_{n-1})\in\mathbb{F}_{q^2}^n$, their Hermitian inner product is defined as $$\langle \textbf x,\textbf y \rangle_h=\sum_{i=0}^{n-1} {x_i}^qy_i={x_0}^qy_0+{x_1}^qy_1+\cdots+{x_{n-1}}^qy_{n-1}.$$
    The Hermitian dual code of $\mathcal{C}$ is defined as
    $$\mathcal{C}^{\perp_h}=\{\textbf x \in\mathbb{F}_{q^2}^n~|~\langle\textbf x,\textbf y \rangle_h=0~{\rm for~ all}~\textbf y\in \mathcal{C}\}.$$
 
    If $\mathcal{C}\subseteq \mathcal{C}^{\perp_h}$, then $\mathcal{C}$ is called a Hermitian self-orthogonal code.
    If $\mathcal{C}^{\perp_h}\subseteq\mathcal{C}$, then $\mathcal{C}$ is called a Hermitian dual-containing code.
  
    Let $\mathcal{C}$ be a linear code over $\mathbb{F}_{q^2}^n$ with length $n$ and $gcd(q,n)=1$.
    If for any codeword  $(c_0,c_1,\cdots,c_{n-1})\in \mathcal{C}$ implies its cyclic shift
    $(c_{n-1},c_0,\cdots,c_{n-2})\in \mathcal{C}$, then $\mathcal{C}$ is said to be a cyclic code. For a cyclic code $\mathcal{C}$, each codeword $c=(c_0,c_1,\cdots,c_{n-1})$ is identified with its polynomial form $c(x)=c_0+c_1x+\dots+c_{n-1}x^{n-1}$ and cyclic code $\mathcal{C}$ of length $n$ is an ideal of $\mathbb{F}_{q^2}[x]/\langle x^n-1 \rangle $. Thus, $\mathcal{C}$ can be generated by a monic polynomial factor of $x^n-1$, i.e. $\mathcal{C}=\langle g(x) \rangle$ and $g(x)|(x^n-1)$. Then $g(x)$ is called the generator polynomial of $\mathcal{C}$ and the dimension of $\mathcal{C}$ is $n-deg(g(x))$.
    
    Let $\lambda$ denote a primitive $n$-th root of unity in some extension field of $\mathbb{F}_{q^2}$. 
    Hence, $ x^n-1=\prod_{i=0}^{n-1}(x-\lambda^i)$. The defining set of cyclic code $\mathcal{C}=\langle g(x)\rangle$ of length $n$ is the set $Z=\{0\leq i\leq n-1~|~g(\lambda^i)=0\}$.
    For $ 0\leq i \leq n-1$, the $q^2$-cyclotomic coset modulo $n$ containing $i$ is defined by the set $$C_i=\{i, iq^2,iq^4,...,iq^{2(m_i-1)} \},$$ where $m_i$ is the smallest positive integer such that $iq^{2m_i}\equiv i~mod~n$. Each $C_i$ corresponds to an irreducible divisor of $x^n-1$ over $\mathbb{F}_{q^2}$. 
    Let $\mathcal{C}$ be an $[n,k,d]$ cyclic code over $\mathbb{F}_{q^2}$ with defining set $Z$.
    Obviously, $Z$ must be a union of some $q^2$-cyclotomic cosets modulo $n$ and $dim(\mathcal{C})=n-|Z|$.
    
    For $ 0\leq b+\delta-2 \leq n-1 $, if $\mathcal{C}$ has defining set $Z_\delta=C_b\bigcup C_{b+1}\bigcup \cdots \bigcup C_{b+\delta-2}$, then it is called cyclic BCH code with designed distance $\delta$. Then the following BCH bound is a lower bound for cyclic codes.
    
    \begin{proposition}\label{pro:2.1}
    \textbf{\textup{\cite{ref27} (BCH Bound for Cyclic Codes)}}~Let $\mathcal{C}$ be a $q^2$-ary cyclic code of length $n$ with defining set $Z$. If $Z$ contains $d-1$ consecutive elements, then the minimum distance of $\mathcal{C}$ is at least $d$.
    \end{proposition}
   
    Let the conjugation transpose of an $m\times n$ matrix $H=(x_{i,j})$ entries in $\mathbb{F}_{q^2}$ is an $n \times m$ matrix $H^\dag=({x_{j,i}^q})$.
 	According to literatures \cite{ref4,ref28}, EAQECCs can be constructed from arbitrary classical linear codes over $\mathbb{F}_{q^2}$, which is given by the following proposition.
 	
    \begin{proposition}\label{pro:2.2}
    If $\mathcal{C}$ is an $[n,k,d]_{q^2}$ classical code and $H$ is its parity check matrix over $\mathbb{F}_{q^2}$, then there exist entanglement-assisted quantum codes with  parameters $[[n,2k-n, d; c]]_q$, where $c = rank(HH^\dag)$. 
    \end{proposition}

   Because it is not easy to determine the number of entangled states $c$ by computing the rank of $HH^\dag$, there are scholars put forward the concept of decomposing the defining set of $\mathcal{C}$ as follows.
 	
    \begin{definition}\label{def:2.3}
	\emph{Let $\mathcal{C}$ be a $q^2$-ary cyclic code of length $n$ with defining set $Z$. Assume
	that $Z_1 = Z \bigcap (-qZ)$ and $Z_2 = Z\backslash Z_1$, where $-qZ = \{n - qx~|~x \in Z\}$. Then,
	$Z = Z_1 \bigcup Z_2$ is called a decomposition of the defining set of $\mathcal{C}$.}
    \end{definition} 

 After decomposing the defining set of $\mathcal{C}$, there is the following lemma which give a relative easy method to calculate the number of entangled states $c$.

    \begin{lemma}\label{le:2.4}
	Let $\mathcal{C}$ be a cyclic code with length n over  $\mathbb{F}_{q^2}$, where $\gcd(n, q) = 1$.
	Suppose that $Z$ is the defining set of the cyclic code $\mathcal{C}$ and $Z = Z_1 \bigcup Z_2$ is a
	decomposition of $Z$. Then, the number of entangled states required is $c = |Z_1|$.	
    \end{lemma}   
 \section{Construction of EAQMDS Codes }

In this section, we devote to derive four new classes of EAQMDS codes from cyclic codes over $\mathbb{F}_{q^2}$. Let $q$ be an odd prime power with $a|(q^2+1)$, where $a$ is even. In this case, we always assume $n=(q^2+1)/a$, where $a=m^2+1$ ($m\geq1$~is~odd), and we consider cyclic codes of length $n$ over $\mathbb{F}_{q^2}$. By Lemma 3.1 in \cite{ref29}, we have the following results directly.

\begin{lemma}\label{le:3.1}
	Let $n=(q^2+1)/a$, where $a=m^2+1$ $(m~is~odd)$, then all cyclotomic cosets modulo $n$ containing $i$ are as follows:
	 $$C_i=\{i,-i\}=\{i,n-i\},$$ for $1\leq i\leq n-1$.
\end{lemma} 

 Next, we give a useful lemma which will be used in later constructions.

\begin{lemma}\label{le:3.2}
	Let $n=(q^2+1)/a,~a=m^2+1 ~(m\geq1~is~odd)$ and $s=(n-1)/2$, $q$ be an odd prime power with the form $a|(q\pm m)$. Let $0\leq l\leq (m-3)/2 $, when $m\geq 3$ and $l=0$, when $m=1$. Then we have the following results in four cases bellow:
	$$-qC_{uq+v}={C_{vq-u}}.$$
	1) When $a|(q-m)$ and $q=2ak+m$, then $1\leq v\leq mk $, if $0\leq u \leq k$, or $mk+1\leq v\leq 2mk $, $2(m+1)k+l(2mk+1)+2\leq v\leq (3m+1)k+l(2mk+1)+1$, $q-l(2mk+1)-(m+1)k \leq v\leq q-l(2mk+1)-(2k+1)$, if $ 0\leq u \leq k-1$.\\\\
	2) When $a|(q-m)$ and $q=2ak+a+m$, then $1\leq v\leq (2k+1)m $, $l(2mk+m+1)+(m+1)(2k+1)+2\leq v\leq l(2mk+m+1)+(2k+1)(3m+1)/2+1$, if $0\leq u \leq k$, or $q-l(2mk+m+1)-(2k+1)(m+1)/2\leq v\leq q-l(2mk+m+1)-2(k+1)$, if $ 0\leq u \leq k-1$.\\\\
	3) When $a|(q+m)$ and $q=2ak+a-m$, then $l(2mk+m-1)+2k+1\leq v\leq l(2mk+m-1)+(2k+1)(m+1)/2-1$, if $0\leq u \leq k$, or $q-l(2mk+m-1)-(2k+1)(3m+1)/2+2\leq v\leq q-l(2mk+m-1)-(m+1)(2k+1)+1$, if $ 0\leq u \leq k-1$.\\\\
	4) When $a|(q+m)$ and $q=2ak+2a-m$, then $l[2m(k+1)-1]+2k+2\leq v\leq [2m(k+1)-1]l+m(k+1)+k $, $q-l[2m(k+1)-1]-3m(k+1)-k+1 \leq v \leq q-l[2m(k+1)-1]-(2m+1)(k+1)-k$, $q-2m(k+1)+1\leq v\leq q-m(k+1)$ if $0\leq u \leq k$, or $q-m(k+1)+1\leq v\leq q $ if $ 0\leq u \leq k-1$.
\end{lemma} 

\begin{proof} 
	1) Note that  $C_{uq+v}=\{uq+v,-(uq+v)\}$  for $0\leq v\leq mk $, if $0\leq u \leq k$, or
    $mk+1\leq v\leq 2mk $, $2(m+1)k+l(2mk+1)+2\leq v\leq (3m+1)k+l(2mk+1)+1$, $q-l(2mk+1)-(m+1)k \leq v\leq q-l(2mk+1)-(2k+1)$, if $ 0\leq u \leq k-1$, where $0\leq l\leq (m-3)/2$, when $m\geq 3$ and $l=0$, when $m=1$.
	
	Since $-q\cdot(-(uq+v))=uq^2+vq=u(q^2+1)+vq-u\equiv vq-u~ mod~n$. This gives that $-qC_{uq+v}={C_{vq-u}}$.
	
	The proofs of 2), 3) and 4) are similar to case 1), so we omit it here.

\end{proof}
    \ 
    
    From Lemma \ref{le:3.2}, we can also obtain $-qC_{jq-t}={C_{tq+j}}$ in each case. Where the range of $q$, $t$ and $j$ is below:
    
    1) When $a|(q-m)$ and $q=2ak+m$, let $0\leq l\leq (m-3)/2$, if $m\geq 3$ and $l=0$, if $m=1$, then $-qC_{j q-t}={C_{tq+j}}$, $0\leq j\leq mk$, if $0\leq t \leq k$, or $mk+1\leq j\leq 2mk$, $2(m+1)k+l(2mk+1)+2\leq j\leq (3m+1)k+l(2mk+1)+1$, $q-l(2mk+1)-(m+1)k \leq j\leq q-l(2mk+1)-(2k+1)$, if $0\leq t \leq k-1$.
    
    In the case 2), 3) and 4), we have similar results.
    
    Based on the discussions above, we can give the first construction as follows.\\\\
\noindent\textbf{Case \uppercase\expandafter{\romannumeral 1} $~~~~\bm {q=2ak+m} $}\\
    
    In order to obtain the number of entangled states $c$, we give the following lemma for preparation.

\begin{lemma}\label{le:3.3}
	Let $n=(q^2+1)/a,~a=m^2+1 ~(m\geq1~is~odd)$, $s=(n-1)/2$ and $q=2ak+m~(k\geq 1)$ be an odd prime power. For a positive integer $1\leq \alpha\leq k$, let 
	\\$$T_{1}=
	\bigcup_{\substack{s+(m+t)k+h+\alpha \leq v \leq s+ (m+t+2)k+(h-1)-\alpha,\\if~ v\leq s+mk,~0 \leq u\leq\alpha,~else~ 0 \leq u\leq\alpha-1\\-m\leq t\leq (2m-1)m~and~t~is~odd}}C_{uq+v}$$
	$$when~ -m\leq t \leq -1,~h=1.$$
	$$when~ 1\leq t\leq 2m-1,~h=2.$$
	$$when~ 2m+1 \leq t \leq 4m-1,~h=3.$$
	$$\dots\dots$$
	$$when~ 2(m-1)m+1 \leq t \leq (2m-1)m,~h=m+1.$$
	Then $T_1\bigcap-qT_1=\emptyset$.
	
\end{lemma}
\begin{proof}
	For a positive integer $\alpha$ with $1 \leq \alpha \leq k$, let 
	\\$$T_{1}=
	\bigcup_{\substack{s+(m+t)k+h+\alpha \leq v \leq s+ (m+t+2)k+(h-1)-\alpha,\\if~ v\leq s+mk,~0 \leq u\leq\alpha,~else~ 0 \leq u\leq\alpha-1\\-m\leq t\leq (2m-1)m~and~t~is~odd}}C_{uq+v}$$
	$$when~ -m\leq t \leq -1,~h=1.$$
	$$when~ 1\leq t\leq 2m-1,~h=2.$$
	$$when~ 2m+1 \leq t \leq 4m-1,~h=3.$$
	$$\dots\dots$$
	$$when~ (2m-2)m+1 \leq t \leq (2m-1)m,~h=m+1.$$
	Then by Lemma \ref{le:3.2}, we have
    \\$$-qT_{1}=
    \bigcup_{\substack{s+(m+t)k+h+\alpha \leq v \leq s+ (m+t+2)k+(h-1)-\alpha,\\if~ v\leq s+mk,~0 \leq u\leq\alpha,~else~ 0 \leq u\leq\alpha-1\\-m\leq t\leq (2m-1)m~and~t~is~odd}}C_{vq-u}$$
    $$when~ -m\leq t \leq -1,~h=1.$$
    $$when~ 1\leq t\leq 2m-1,~h=2.$$
    $$when~ 2m+1 \leq t \leq 4m-1,~h=3.$$
    $$\dots\dots$$
    $$when~ (2m-2)m+1 \leq t \leq (2m-1)m,~h=m+1.$$
	When $-m\leq t_1 \leq -1,~h_1=1$, then $s+1+\alpha \leq v_1\leq s+mk$ and $0\leq u_1\leq \alpha$, it follows that
	$$u_1q+v_1\leq \alpha q+s+mk,~(s+1+\alpha)q-\alpha \leq v_1q-u_1.$$
	When $-m\leq t_2 \leq -1,~h_2=1$, then $s+mk+1 \leq v_2\leq s+(m+1)k-\alpha$ and $0\leq u_2\leq \alpha-1$, it follows that
	$$u_2q+v_2\leq (\alpha-1)q+s+(m+1)k-\alpha,~(s+mk+1)q-\alpha+1 \leq v_2q-u_2.$$
	When $1\leq t_3 \leq 2m-1,~h_3=2$, then $s+(m+1)k+2+\alpha \leq v_3\leq s+(3m+1)k-\alpha+1$ and $0\leq u_3\leq \alpha-1$, it follows that
    $$u_3q+v_3\leq (\alpha-1)q+s+(3m+1)k-\alpha+1,~[s+(m+1)k+2+\alpha]q-\alpha+1 \leq v_3q-u_3.$$
	When $2m+1\leq t_4 \leq 4m-1,~h_4=3$, then $s+(3m+1)k+3+\alpha \leq v_4\leq s+(5m+1)k+2-\alpha$ and $0\leq u_4\leq \alpha-1$, it follows that
    $$u_4q+v_4\leq (\alpha-1)q+s+(5m+1)k+2-\alpha,~[s+(3m+1)k+3+\alpha]q-\alpha+1 \leq v_4q-u_4.$$
    $$\dots\dots$$
    When $2(m-1)m+1\leq t_{m+2} \leq (2m-1)m,~h_{m+2}=m+1$, then $s+[(2m-1)m+1]k+m+1+\alpha \leq v_{m+2}\leq s+2(m^2+1)k+m-\alpha$ and $0\leq u_{m+2}\leq \alpha-1$, it follows that
    $$u_{m+2}q+v_{m+2}\leq (\alpha-1)q+s+2(m^2+1)k+m-\alpha,$$ $$[s+(2m^2-m+1)k+m+1+\alpha ]q-\alpha+1 \leq v_{m+2}q-u_{m+2}.$$\\
	It is easy to check that\\
	$u_1q+v_1<v_1q-u_1,~u_1q+v_1<v_2q-u_2,~\dots\dots~,~u_1q+v_1<v_{m+2}q+u_{m+2}.$
	$u_2q+v_2<v_1q-u_1,~u_2q+v_2<v_2q-u_2,~\dots\dots~,~u_2q+v_2<v_{m+2}q+u_{m+2}.$
	$$\dots\dots$$
	$u_{m+2}q+v_{m+2}<v_1q-u_1,~\dots\dots~,~u_{m+2}q+v_{m+2}<v_{m+2}q+u_{m+2}.$
	
	For the range of $~v_1,~v_2,~\dots\dots~,~v_{m+2}$ and $~u_1,~u_2,~\dots\dots~,~u_{m+2}$, note that
	$u_iq+v_i\leq (m+q)k~(i=1,2,\dots\dots,m+2) $, the subscripts of $C_{u_i+v_i}$ is the biggest number in the set. 
	Then $T_1\bigcap-qT_1=\emptyset$. The desired results follows.
	
\end{proof}
\
    
    Based on Lemma \ref{le:3.3}, we can determine the number of entangled states $c$ in the following theorem.
    
\begin{theorem}\label{th:3.4}
    Let $n=(q^2+1)/a$, $a=m^2+1$ $(m\geq 1~is~odd)$ and $q=2ak+m$ $(k\geq 1)$ be an odd prime power. For a positive integer $\alpha$ with $1\leq \alpha\leq k$, 
	let $\mathcal{C}$ be a cyclic code with defining set $Z$ given as follows 
	$$Z=C_{s+1}\bigcup C_{s+2}\bigcup\dots\bigcup C_{s+(\alpha q+mk)}.$$
	Then $|Z_{1}|=4\alpha(a\alpha +m)$.
\end{theorem}
\begin{proof}
	Let $$T_{1}=
	\bigcup_{\substack{s+(m+t)k+h+\alpha \leq v \leq s+ (m+t+2)k+(h-1)-\alpha,\\if~ v\leq s+mk,~0 \leq u\leq\alpha,~else~ 0 \leq u\leq\alpha-1\\-m\leq t\leq (2m-1)m~and~t~is~odd}}C_{uq+v}$$
	$$when~ -m\leq t \leq -1,~h=1.$$
	$$when~ 1\leq t\leq 2m-1,~h=2.$$
	$$when~ 2m+1 \leq t \leq 4m-1,~h=3.$$
	$$\dots\dots$$
	$$when~ (2m-2)m+1 \leq t \leq (2m-1)m,~h=m+1.$$
	and
	
	\begin{equation*}
    \begin{split}
	T_1'~~~~=&~~~~~~~\bigcup_{\substack{s+1\leq v \leq s+\alpha,~0\leq u\leq\alpha}}C_{uq+v}
	~~~~~~~~~~\bigcup_{\substack{s+2tk+1-\alpha\leq v \leq s+2tk+\alpha,\\1\leq t\leq \frac{m-1}{2},~0\leq u\leq \alpha}}C_{uq+v}\\
	&\bigcup_{\substack{s+2tk+\frac{f+3}{2}-\alpha\leq v \leq s+2tk+\frac{f+1}{2}+\alpha,\\\frac{fm+3}{2}\leq t \leq \frac{(f+2)m-1}{2},\\1\leq f\leq 2m-1~is~odd,~0\leq u\leq \alpha-1}}C_{uq+v}
    ~\bigcup_{\substack{s+2tk+m+1-\alpha\leq v \leq s+2tk+m+\alpha,\\\frac{(2m-1)m+3}{2}\leq t \leq m^2,~0\leq u\leq \alpha-1}}C_{uq+v}\\
	&~~~~~\bigcup_{\substack{s+2ak+m+1-\alpha\leq v \leq s+q,\\0\leq u\leq \alpha-1}}C_{uq+v}
	\bigcup_{\substack{s+(gm+1)k+\frac{g+1}{2}-\alpha\leq v \leq s+(gm+1)k+\frac{g+1}{2}+\alpha,\\1\leq g \leq 2m-1~is~odd,~0\leq u\leq \alpha-1}}C_{uq+v}.
    \end{split}
    \end{equation*}

	From Lemma \ref{le:3.2}, we have 
	
    \begin{equation*}
    \begin{split}
    -qT_1'~~~=&~~~~~~~\bigcup_{\substack{s+1\leq v \leq s+\alpha,~0\leq u\leq\alpha}}C_{vq-u}
    ~~~~~~~~~~~\bigcup_{\substack{s+2tk+1-\alpha\leq v \leq s+2tk+\alpha,\\1\leq t\leq \frac{m-1}{2},~0\leq u\leq \alpha}}C_{vq-u}\\
    &\bigcup_{\substack{s+2tk+\frac{f+3}{2}-\alpha\leq v \leq s+2tk+\frac{f+1}{2}+\alpha,\\\frac{fm+3}{2}\leq t \leq \frac{(f+2)m-1}{2},\\1\leq f\leq 2m-1~is~odd,~0\leq u\leq \alpha-1}}C_{vq-u}
    ~\bigcup_{\substack{s+2tk+m+1-\alpha\leq v \leq s+2tk+m+\alpha,\\ \frac{(2m-1)m+3}{2}\leq t \leq m^2,~0\leq u\leq \alpha-1}}C_{vq-u}\\
    &~~~~~\bigcup_{\substack{s+2ak+m+1-\alpha\leq v \leq s+q,\\0\leq u\leq \alpha-1}}C_{vq-u}
    \bigcup_{\substack{s+(gm+1)k+\frac{g+1}{2}-\alpha\leq v \leq s+(gm+1)k+\frac{g+1}{2}+\alpha,\\1\leq g \leq 2m-1~is~odd,~0\leq u\leq \alpha-1}}C_{vq-u}.
    \end{split}
    \end{equation*}
	
	It is easy to check that $-qT_1'=T_1'$. From the definitions of $Z$, $T_1$ and $T_1'$, we have $Z=T_1\bigcup T_1'$. 
	Then from the definition of $Z_{1}$,
\begin{equation*}
\begin{split}
	Z_{1}=Z\bigcap (-qZ)&=(T_1\bigcup T_1')\bigcap(-qT_1\bigcup -qT_1')\\
	&=(T_1\bigcap-qT_1)\bigcup(T_1\bigcap-qT_1')\bigcup(T_1'\bigcap-qT_1)\bigcup(T_1'\bigcap-qT_1')\\
	&=T_1'.
\end{split}
\end{equation*}
Therefore, $|Z_{1}|=|T_1'|=4\alpha(a\alpha +m)$.
	
\end{proof}
    \ 
    
    From Lemmas \ref{le:3.2}, \ref{le:3.3} and Theorem \ref{th:3.4} above, we can obtain the first construction of EAQMDS codes in the following theorem.
    
\begin{theorem}\label{th:3.5}
	Let $n=(q^2+1)/a$, $a=m^2+1$ $(m\geq1~is~odd)$ and $q=2ak+m$ $(k\geq 1)$ be an odd prime power. There are EAQMDS codes with parameters 
    $$[[n, n-4\alpha(q-m-a\alpha)-4mk, 2(\alpha q+mk)+1; 4\alpha(a\alpha+m)]]_q,$$ where $1\leq \alpha\leq k$.
\end{theorem}
\begin{proof}
	For a positive integer $1\leq \alpha\leq k$.
	Suppose that $\mathcal{C}$ is a cyclic code of length $n=(q^2+1)/a,~a=m^2+1$ ($m\geq1$~is~odd) with defining set 
	$$Z=C_{s+1}\bigcup C_{s+2}\bigcup\dots\bigcup C_{s+(\alpha q+mk)}.$$
	Then the dimension of $\mathcal{C}$ is $n-2(\alpha q+mk)$. Note that cyclic code $\mathcal{C}$ have $2(\alpha q+mk)$ consecutive roots.
	By Proposition \ref{pro:2.1}, the minimum distance of $\mathcal{C}$ is at least $2(\alpha q+mk)+1$. From the Singleton bound, $\mathcal{C}$ is an MDS code with parameters $[[n,n-2(\alpha q+mk), 2(\alpha q+mk)+1]]_{q^2}$.
	Then by Theorem \ref{th:3.4}, we have $|Z_{1}|=|T_1'|=4\alpha(a\alpha +m)$.
   From Proposition \ref{pro:2.2} and Lemma \ref{le:2.4}, there are EAQECCs with parameters
    $$[[n, n-4\alpha(q-m-a\alpha)-4mk, 2(\alpha q+mk)+1; 4\alpha(a\alpha+m)]]_q.$$
    It is easy to check that $$n-k+c+2=4(\alpha q+mk)+2=2d.$$
    By Proposition \ref{pro:1.1}, it implies that the EAQECCs we constructed are EAQMDS codes.
    
\end{proof}
    \
    
    \begin{table}[H]
	\centering
	\renewcommand{\tablename}
	\caption{ {TABLE $\rm\uppercase\expandafter{\romannumeral 1}$:} SAMPLE PARAMETERS OF MDS \protect\\ EAQECCs OF THEOREM 3.5 }\\
	\begin{tabular}{c<{\centering}c<{\centering}c<{\centering}c<{\centering}c<{\centering}}\\
		\hline
		$m$ & $q$  & $n$  &$\alpha$  &   Parameters
		\\
		\Xhline{1.0pt}
		
		$1$ & $13$  & $85$ &  $1$  & $[[85,33,33;12]]_{13}$ \\
		$1$ & $13$  & $85$ &  $2$  & $[[85,9,59;40]]_{13}$ \\
		$1$ & $13$  & $85$ &  $3$  & $[[85,1,85;84]]_{13}$ \\
		$1$ & $17$  & $145$ &  $1$  & $[[145,73,43;12]]_{17}$ \\
		$1$ & $17$  & $145$ &  $2$  & $[[145,33,77;40]]_{17}$ \\
		$1$ & $17$  & $145$ &  $3$  & $[[145,9,111;84]]_{17}$ \\
		$1$ & $17$  & $145$ &  $4$  & $[[145,1,145;144]]_{17}$ \\
		$3$ & $43$  & $185$ &  $1$  & $[[185,41,99;52]]_{43}$ \\
		$3$ & $43$  & $185$ &  $2$  & $[[185,1,185;184]]_{43}$ \\
		$3$ & $83$  & $689$ &  $1$  & $[[689,361,191;52]]_{83}$ \\
		$3$ & $83$  & $689$ &  $2$  & $[[689,161,357;184]]_{83}$ \\
		$3$ & $83$  & $689$ &  $3$  & $[[689,41,523;396]]_{83}$ \\
		$3$ & $83$  & $689$ &  $4$  & $[[689,1,689;688]]_{83}$ \\
	
		$5$ & $109$  & $457$ &  $1$  & $[[457,105,239;124]]_{109}$ \\
		$5$ & $109$  & $457$ &  $2$  & $[[457,1,457;456]]_{109}$ \\
		%  $5$ & $317$  & $3865$ &  $1$  & $[[3865,2601,695;124]]_{317}$ \\
		%   $5$ & $317$  & $3865$ &  $2$  & $[[3865,1665,1329;465]]_{317}$ \\
	% 	$5$ & $317$  & $3865$ &  $3$  & $[[3865,937,1963;996]]_{317}$ \\
	% 	$5$ & $317$  & $3865$ &  $4$  & $[[3865,417,2597;1744]]_{317}$ \\
	% 	$5$ & $317$  & $3865$ &  $5$  & $[[3865,105,3231;2700]]_{317}$ \\
	% 	$5$ & $317$  & $3865$ &  $6$  & $[[3865,1,3865;3864]]_{317}$ \\	
		\hline		
	\end{tabular}
\end{table} 
\

\noindent\textbf{Case \uppercase\expandafter{\romannumeral 2} $~~~~\bm {q=2ak+a+m} $}\\

    As for the case that $n=(q^2+1)/a,~a=m^2+1$ ($m\geq1$~is~odd) and $q=2ak+a+m$ $(k\geq 1)$ is an odd prime power, we can produce the following EAQMDS codes. The proofs are similar to that in the Case \uppercase\expandafter{\romannumeral 1}, so we omit it here.
    
\begin{lemma}\label{le:7}

    Let $n=(q^2+1)/a$,~$a=m^2+1$ $(m\geq1~is~odd)$, $s=(n-1)/2$ and $q=2ak+a+m~(k\geq 1)$ be an odd prime power. For a positive integer $1\leq \alpha\leq k$, let 
   
    \begin{equation*}
    \begin{split}
    T_{1}=
    \bigcup_{\substack{s+t(2k+1)+(h+1)+\alpha \leq v \leq s+ (t+1)(2k+1)+(h-1)-\alpha,\\if~ v\leq s+(a+m)k+\frac{a+2m}{2},~0 \leq u\leq\alpha,~else~ 0 \leq u\leq\alpha-1\\(h-1)m+\gamma_1\leq t\leq  hm-\gamma_2,~1\leq h\leq m}}C_{uq+v}\\
    when~ 1\leq h \leq (m-1)/2, \gamma_1=0, \gamma_2=1;\\
    when~h = (m+1)/2, \gamma_1=0, \gamma_2=0;\\
    when~ (m+1)/2+1 \leq h \leq m, \gamma_1=1, \gamma_2=0.\\
    \end{split}
    \end{equation*}

    Then $T_1\bigcap-qT_1=\emptyset$.
	
\end{lemma}

\begin{theorem}\label{th:8}
	Let $n=(q^2+1)/a$,~$a=m^2+1$ $(m\geq1~is~odd)$, $s=(n-1)/2$ and $q=2ak+a+m~(k\geq 1)$ be an odd prime power. For a positive integer m with $1\leq \alpha\leq k$, 
	let $\mathcal{C}$ be a cyclic code with defining set $Z$ given as follows 
	$$Z=C_{s+1}\bigcup C_{s+2}\bigcup\dots\bigcup C_{s+[\alpha q +(a+m)k+\frac{a+2m}{2}]}.$$
	Then $|Z_{1}|=4\alpha(a\alpha+a+m)+a+2m$.
\end{theorem}

\begin{theorem}\label{th:9}
    Let $n=(q^2+1)/a1$,~$a=m^2+1$ $(m\geq1~is~odd)$, $s=(n-1)/2$ and $q=2ak+a+m~(k\geq 1)$ be an odd prime power. There are EAQMDS codes with parameters 
    $$[[n, n-4\alpha(q-a-m-a\alpha)-4(a+m)k-(a+2m),$$
    $$ 2[\alpha q +(a+m)k+\frac{a+2m}{2}]+1; 4\alpha(a\alpha+a+m)+a+2m]]_q.$$ where $1\leq \alpha\leq k$.
	
\end{theorem}

    \begin{table}[H]
    	\centering
    	\renewcommand{\tablename}
    	\caption{ {TABLE $\rm\uppercase\expandafter{\romannumeral 2}$:} SAMPLE PARAMETERS OF MDS \protect\\ EAQECCs OF THEOREM 3.8 }\\
    	\begin{tabular}{c<{\centering}c<{\centering}c<{\centering}c<{\centering}c<{\centering}}\\
    		\hline
    		$m$ & $q$  & $n$  &$\alpha$  &   Parameters
    		\\
    		\Xhline{1.0pt}
    		
    		$1$ & $11$  & $61$ &  $1$  & $[[61,9,39;24]]_{11}$ \\
    		$1$ & $11$  & $61$ &  $2$  & $[[61,1,61;60]]_{11}$ \\
    	    $1$ & $19$  & $181$ &  $1$  & $[[181,73,67;24]]_{19}$ \\
    	    $1$ & $19$  & $181$ &  $2$  & $[[181,33,105;60]]_{19}$ \\
    	    $1$ & $19$  & $181$ &  $3$  & $[[181,9,143;112]]_{19}$ \\
    	    $1$ & $19$  & $181$ &  $4$  & $[[181,1,181;180]]_{19}$ \\
    	    $3$ & $53$  & $281$ &  $1$  & $[[281,41,175;108]]_{53}$ \\
    	    $3$ & $53$  & $281$ &  $2$  & $[[281,1,281;280]]_{53}$ \\
    	    $3$ & $73$  & $533$ &  $1$  & $[[533,161,241;108]]_{73}$ \\
    	    $3$ & $73$  & $533$ &  $2$  & $[[533,41,387;280]]_{73}$ \\
    	    $3$ & $73$  & $533$ &  $3$  & $[[533,1,533;532]]_{73}$ \\
    	    $5$ & $239$  & $2197$ &  $1$  & $[[2197,937,763;264]]_{239}$ \\
    	    $5$ & $239$  & $2197$ &  $2$  & $[[2197,417,1241;700]]_{239}$ \\
    	    $5$ & $239$  & $2197$ &  $3$  & $[[2197,105,1719;1344]]_{239}$ \\
    	    $5$ & $239$  & $2197$ &  $4$  & $[[2197,1,2197;2196]]_{239}$ \\
    		\hline		
    	\end{tabular}
    \end{table}

\noindent\textbf{Case \uppercase\expandafter{\romannumeral 3} $~~~~\bm{q=2ak+a-m} $}\\

    we also have similar results for $n=(q^2+1)/a,~a=m^2+1$ ($m\geq1$~is~odd) and $q=2ak+a-m$ $(k\geq 1)$ is an odd prime power, we can produce the following EAQMDS codes. These results are given in the following lemma and theorems. Because the proofs of them are similar to that in Lemma \ref{le:3.3} and Theorems \ref{th:3.4}, \ref{th:3.5}, we omit it here.

\begin{lemma}\label{le:10}

	Let $n=(q^2+1)/a$,~$a=m^2+1$ $(m\geq1~is~odd)$, $s=(n-1)/2$ and $q=2ak+a-m~(k\geq 1)$ be an odd prime power. For a positive integer $1\leq \alpha\leq k$, let 
	\\$$T_{1}=
	\bigcup_{\substack{s+t(2k+1)+(2-h)+\alpha \leq v \leq s+t(2k+1)+2k-(1-h)-\alpha,\\if~ v\leq s+(a-m)k+\frac{a-2m}{2},~0 \leq u\leq\alpha,~else~ 0 \leq u\leq\alpha-1\\(h-1)m+\gamma_1\leq t\leq  hm-\gamma_2,~1\leq h\leq m}}C_{uq+v}
	$$
	$$when~ 1\leq h \leq (m-1)/2,~\gamma_1=0,~\gamma_2=1;$$
	$$when~ h = (m+1)/2,~\gamma_1=0,~\gamma_2=0;$$
	$$when~ (m+1)/2+1 \leq h \leq m,~\gamma_1=1,~\gamma_2=0.$$
	Then $T_1\bigcap-qT_1=\emptyset$.
		
\end{lemma}
	
\begin{theorem}\label{th:11}
	Let $n=(q^2+1)/a$, $a=m^2+1$ $(m\geq1~is~odd)$, $s=(n-1)/2$ and $q=2ak+a-m~(k\geq 1)$ be an odd prime power. For a positive integer m with $1\leq \alpha\leq k$, 
	let $\mathcal{C}$ be a cyclic code with defining set $Z$ given as follows 
	$$Z=C_{s+1}\bigcup C_{s+2}\bigcup\dots\bigcup C_{s+[\alpha q +(a-m)k+\frac{a-2m}{2}]}.$$
	Then $|Z_{1}|=4\alpha(a\alpha+a-m)+a-2m$.
\end{theorem}
	
\begin{theorem}\label{th:12}
	Let $n=(q^2+1)/a$, $a=m^2+1$ $(m\geq1~is~odd)$, $s=(n-1)/2$ and $q=2ak+a-m~(k\geq 1)$ be an odd prime power. There are EAQMDS codes with parameters 
	$$[[n, n-4\alpha(q-a+m-a\alpha)-4(a-m)k-(a-2m),$$ $$2[\alpha q +(a-m)k+\frac{a-2m}{2}]+1;
	4\alpha(a\alpha+a-m)+a-2m]]_q.$$ where $1\leq \alpha\leq k$.
		
\end{theorem}

 \begin{table}[H]
	\centering
	\renewcommand{\tablename}
	\caption{ {TABLE $\rm\uppercase\expandafter{\romannumeral 3}$:} SAMPLE PARAMETERS OF MDS \protect\\ EAQECCs OF THEOREM 3.11 }\\
	\begin{tabular}{c<{\centering}c<{\centering}c<{\centering}c<{\centering}c<{\centering}}\\
		\hline
		$m$ & $q$  & $n$  &$\alpha$  &   Parameters
		\\
		\Xhline{1.0pt}
		
		$1$ & $29$  & $421$ &  $1$  & $[[421,289,73;12]]_{29}$ \\
		$1$ & $29$  & $421$ &  $2$  & $[[421,201,131;40]]_{29}$ \\
		$1$ & $29$  & $421$ &  $3$  & $[[421,129,189;84]]_{29}$ \\
		$1$ & $29$  & $421$ &  $4$  & $[[421,73,247;144]]_{29}$ \\
		$1$ & $29$  & $421$ &  $5$  & $[[421,33,305;220]]_{29}$ \\
		$1$ & $29$  & $421$ &  $6$  & $[[421,9,363;312]]_{29}$ \\
		$1$ & $29$  & $421$ &  $7$  & $[[421,1,421;420]]_{29}$ \\
		
		$3$ & $47$  & $221$ &  $1$  & $[[221,41,127;72]]_{47}$ \\
		$3$ & $47$  & $221$ &  $2$  & $[[221,1,221;220]]_{47}$ \\
		
		$3$ & $67$  & $449$ &  $1$  & $[[449,161,181;72]]_{47}$ \\
		$3$ & $67$  & $449$ &  $2$  & $[[449,41,315;220]]_{47}$ \\
		$3$ & $67$  & $449$ &  $3$  & $[[449,1,449;448]]_{47}$ \\
		
		$5$ & $229$  & $2017$ &  $1$  & $[[2017,937,643;204]]_{229}$ \\
		$5$ & $229$  & $2017$ &  $2$  & $[[2017,417,1101;600]]_{229}$ \\
		$5$ & $229$  & $2017$ &  $3$  & $[[2017,105,1559;1204]]_{229}$ \\
		$5$ & $229$  & $2017$ &  $4$  & $[[2017,1,2017;2016]]_{229}$ \\
	
		\hline		
	\end{tabular}
\end{table}

\noindent\textbf{Case \uppercase\expandafter{\romannumeral 4} $\bm{~~~~q=2ak+2a-m}$}\\

   Similarly, for the case that $n=(q^2+1)/a$,~$a=m^2+1$ ($m\geq1$~is~odd) and $q=2ak+2a-m$ $(k\geq 1)$ is an odd prime power, there are EAQMDS codes as follows. We only list the results in the following lemma and theorems and omit proofs for simplification.
    
\begin{lemma}\label{le:13}
	Let $n=(q^2+1)/a$,~$a=m^2+1$ $(m\geq1~is~odd)$, $s=(n-1)/2$ and $q=2ak+2a-m~(k\geq 1)$ be an odd prime power. For a positive integer $1\leq \alpha\leq k$, let 
	\\$$T_{1}=
	\bigcup_{\substack{s+(m+t)(k+1)+h+\alpha \leq v \leq s+ (m+t+2)(k+1)+(h-3)-\alpha,\\if~ v\leq s+(2a-m)k+2(a-m),~0 \leq u\leq\alpha,~else~ 0 \leq u\leq\alpha-1\\-m\leq t\leq (2m-1)m~and~t~is~odd}}C_{uq+v}
	$$
	$$when~ -m\leq t \leq -1,~h=2.$$
	$$when~ 1\leq t\leq 2m-1,~h=1.$$
	$$when~ 2m+1 \leq t \leq 4m-1,~h=0.$$
	$$\dots\dots$$
	$$when~ (2m-2)m+1 \leq t \leq (2m-1)m,~h=2-m.$$
	Then $T_1\bigcap-qT_1=\emptyset$.
	
\end{lemma}

\begin{theorem}\label{th:14}
	Let $n=(q^2+1)/a$, $a=m^2+1$ $(m\geq1~is~odd)$ and $q=2ak+2a-m$ $(k\geq 1)$ be an odd prime power. For a positive integer m with $1\leq \alpha\leq k$, 
	let $\mathcal{C}$ be a cyclic code with defining set $Z$ given as follows 
	$$Z=C_{s+1}\bigcup C_{s+2}\bigcup\dots\bigcup C_{s+[\alpha q+(2a-m)k+2(a-m)]}.$$
	Then $|Z_{1}|=4\alpha(a\alpha +2a-m)+4(a-m)$.
\end{theorem}

\begin{theorem}\label{th:15}
	Let $n=(q^2+1)/a$, $a=m^2+1$ $(m\geq1~is~odd$) and $q=2ak+2a-m$ $(k\geq 1)$ be an odd prime power. There are EAQMDS codes with parameters 
	$$[[n, n-4\alpha[q-(2a-m)-a\alpha]-4(2a-m)k-4(a-m),$$
	$$ 2[\alpha q+(2a-m)k+2(a-m)]+1; 4\alpha(a\alpha +2a-m)+4(a-m)]]_q,$$ where $1\leq \alpha\leq k$.
	
\end{theorem}
    
    \begin{table}[H]
    \centering
    \renewcommand{\tablename}
    \caption{ {TABLE $\rm\uppercase\expandafter{\romannumeral 4}$:} SAMPLE PARAMETERS OF MDS \protect\\ EAQECCs OF THEOREM 3.14 }\\
    \begin{tabular}{c<{\centering}c<{\centering}c<{\centering}c<{\centering}c<{\centering}}\\
    \hline
    $m$ & $q$  & $n$  &$\alpha$  &   Parameters\\
    \Xhline{1.0pt}
    $1$ & $23$  & $265$ &  $1$  & $[[265,129,81;24]]_{23}$ \\
    $1$ & $23$  & $265$ &  $2$  & $[[265,73,127;60]]_{23}$ \\
    $1$ & $23$  & $265$ &  $3$  & $[[265,33,173;112]]_{23}$ \\
    $1$ & $23$  & $265$ &  $4$  & $[[265,9,219;180]]_{23}$ \\
    $1$ & $23$  & $265$ &  $5$  & $[[265,1,265;264]]_{23}$ \\
    
    $3$ & $97$  & $941$ &  $1$  & $[[941,361,359;136]]_{97}$ \\
    $3$ & $97$  & $941$ &  $2$  & $[[941,161,553;324]]_{97}$ \\
    $3$ & $97$  & $941$ &  $3$  & $[[941,41,747;592]]_{97}$ \\
    $3$ & $97$  & $941$ &  $4$  & $[[941,1,941;940]]_{97}$ \\
    
    $5$ & $151$  & $877$ &  $1$  &$[[877,105,575;376]]_{151}$ \\
    $5$ & $151$  & $877$ &  $2$  &$[[877,1,877;876]]_{151}$ \\
   
    		\hline		
    	\end{tabular}
    \end{table} 

\noindent\textbf{Remark}
~The required number of entangled
states $c$ of the EAQMDS codes obtained in the literatures (see for
instance \cite{ref13},  \cite{ref14},  \cite{ref15},  \cite{ref16}) is fixed. However, the EAQMDS codes we constructed in Theorems 3.5, 3.8, 3.11 and 3.14 with flexible parameters. And we list the parameters of these codes
 in Table $\uppercase\expandafter{\romannumeral 1}$, $\uppercase\expandafter{\romannumeral 2}$,  $\uppercase\expandafter{\romannumeral 3}$ and $\uppercase\expandafter{\romannumeral 4}$.

\section{Conclusion}
    In this paper, we have utilized decomposing the defining set of cyclic codes to extend the present results about EAQMDS codes and construct series of EAQMDS codes with good parameters of length $n=(q^2+1)/a$. They conclude all the length with the form $n=(q^2+1)/a$, where $a=m^2+1$, $m\geq1$~is~odd and $q$ is an odd prime power with the form of $a|(q+m)$ or $a|(q-m)$. Compared with known results, these EAQMDS codes have flexible parameters and much bigger minimum distance than the known quantum MDS codes and EAQECCs with the same length. Therefore, our EAQMDS codes can detect and correct more errors. The study of EAQECCs is an interesting problem in coding theory. We believe that it will have tremendous application in the future.

\dse{Acknowledgments}
    This research is supported by the National Natural Science Foundation of China (No.61772168).

\end{document}